\documentclass{amsart}

\usepackage{amsmath, amsthm}
\usepackage{amsfonts,amssymb}
\usepackage{float}
\usepackage{graphics}
\usepackage{graphicx}
\usepackage[thinlines,thiklines]{easybmat}
\usepackage{multirow}
\usepackage{algorithmic}
\usepackage{algorithm}

\newcommand{\Z}{\mathbb{Z}} % les entiers relatifs
\newcommand{\Q}{\mathbb{Q}} % les rationnels
\newcommand{\R}{\mathbb{R}} % les rationnels
\newcommand{\F}{\mathbb{F}} % un corps fini
\newcommand{\K}{\mathbb{K}} % un corps
\newcommand{\p}{\mathfrak{p}} % un p gothik
\newcommand{\M}{\mathcal{M}}
\newcommand{\Cl}{\operatorname{Cl}}
\newcommand{\OO}{\mathcal{O}}

\newtheorem{theorem}{Theorem}

\newtheorem{lemma}[theorem]{Lemma}
\newtheorem{heuristic}[theorem]{Heuristic}
\newtheorem{proposition}[theorem]{Proposition}
\newtheorem{corollary}[theorem]{Corollary}
\newtheorem{notation}[theorem]{Notation}

\begin{document}
\title[An $L(1/3)$ algorithm for number fields]{An $L(1/3)$ algorithm for ideal class group and regulator computation in certain number fields}     % écrit le titre
\author{Jean-Fran\c{c}ois Biasse}
\address{LIX , \'{E}cole Polytechnique , 91128 PALAISEAU , France }
\email{biasse@lix.polytechnique.fr}
\thanks{The author was supported by a DGA grant}
\subjclass[2000]{Primary 54C40, 14E20; Secondary 46E25, 20C20}
\keywords{Number fields, ideal class group, regulator, units, index calculus, subexponentiality}
%\tableofcontents      % écrit la table des matières

\begin{abstract}
We analyse the complexity of the computation of the class group structure, regulator, and a system of fundamental units of a certain class of number fields. Our approach differs from Buchmann's, who proved a complexity bound of $L(1/2,O(1))$ when the discriminant tends to infinity with fixed degree. We achieve a subexponential complexity in $O(L(1/3,O(1)))$ when both the discriminant and the degree of the extension tend to infinity by using techniques due to Enge and Gaudry in the context of algebraic curves over finite fields. 
\end{abstract}

\maketitle 

\section{Introduction}

Let $\K=\Q(\theta)$ be a number field of degree $n$ and discriminant $\Delta$. The ideal class group of its maximal order $\OO_{\K}$ is a finite abelian group that can be decomposed as: 
$$\Cl(\OO_{\K})=\bigoplus_i \Z/d_i\Z,$$
with $d_i\mid d_{i+1}$. Computing the structure of $\Cl(\OO_{\K})$, along with the regulator and a system of fundamental units of $\OO_{\K}$ is a major task in computational number theory. In addition, many algorithms solving the discrete logarithm problem are based on the group structure computation. 

In 1968, Shanks \cite{Shanks,Shanks2} proposed an algorithm relying on the baby-step giant-step method to compute the structure of the ideal class group and the regulator of a quadratic number field in time $O\left( |\Delta |^{1/4 + \epsilon} \right)$, or $O\left( |\Delta |^{1/5 + \epsilon} \right)$ under the extended Riemann hypothesis \cite{LenstraShanks}. Then, a subexponential strategy for the computation of the group structure of the class group of an imaginary quadratic extension was described in 1989 by Hafner and McCurley \cite{hafner}. The expected running time of this method is
%$$L\left( \sqrt{2}\right) = \left( \exp \sqrt{\log\Delta\log\log\Delta} \right)^{\sqrt{2}+o(1)}$$
$$L_{\Delta}(1/2,\sqrt{2}+o(1)) = e^{ \left( \sqrt{2}+o(1)\right) \sqrt{\log|\Delta|\log\log|\Delta|} }.$$
Buchmann \cite{Buchmann} generalized this result to the case of an arbitrary extension, the complexity being valid for fixed degree $n$ and $\Delta$ tending to infinity. Enge \cite{Enge_L12} used this technique in the context of discrete logarithm computations in the Jacobian of hyperelliptic curves, and developed with Gaudry \cite{Enge} an algorithm for computing the group structure of the Jacobian and solving the discrete logarithm problem for a certain class of curves in time:
$$L_{q^g}\left( 1/3 , O(1) \right) = e^{O(1)\left( \log(q^g)^{1/3} \log\log(q^g)^{2/3}\right)} .$$
In this paper, we adapt the $L(1/3)$ algorithm of Enge and Gaudry to the computation of the group structure of the ideal class group, the regulator, and a system of fundamental units of $\OO_{\K}$. We deal with the case where both the discriminant and the degree of the extension grow to infinity in certain proportions, whereas in \cite{Buchmann} the degree is assumed to be fixed.

\section{Main idea}

We consider a number field $\K = \Q(\theta)$ of discriminant $\Delta$ which can be written as:
$$\K = \Q[X]/T(X),$$
with $T(X) = t_nX^n + t_{n-1}X^{n-1} + \hdots + t_0\in\Z[X]$, and $n:=[\K:\Q]$. Let $d$ be a bound on the bit size of the coefficients of $T$: 
$$d := \max_i \left\lbrace \log(t_i)\right\rbrace.$$ 
In addition, we require that:
\begin{align}\label{cond_n}
 &n\leq n_0\log\left( |\Delta|\right)^{\alpha}(1+o(1))\\
& d\leq d_0\log\left( |\Delta|\right)^{1-\alpha}(1+o(1)),\label{cond_d}
\end{align}
for some $\alpha\in\left[ \frac{1}{3},\frac{2}{3}\right[$, and  some constants $n_0$ and $d_0$. We define  $\kappa:=n_0d_0$. We also denote by $r_1$ the number of real places, by $r_2$ the number of complex places and we define $r:=r_1+r_2-1$. Our algorithm computes the group structure of $\Cl(\Z[\theta])$, its regulator, and a system of fundamental units of $\Z[\theta]$, in expected time lying in:
$$O\left( L_{\text{Disc}(T)}(1/3 , O(1)) \right).$$
In the case of number fields satisfying $\Z[\theta] = \OO_{\K}$ and the above restrictions, we compute the group structure of $\Cl(\OO_{\K})$, $R$, and a system of fundamental units, in expected time $L_{\Delta}(1/3,O(1))$. From now on, we assume that $\K$ satisfies \eqref{cond_n}, \eqref{cond_d}, and $\Z[\theta]=\OO_{\K}$.

\subsubsection*{Example}

Let $\Delta\in\Z$, and $\K_{n,K}$ be an extension of $\Q$ defined by an irreducible polynomial of the form: 
$$T(X) = X^n - K,$$
with 
\begin{align*}
&\log K = \left\lfloor \log\left( |\Delta|\right)^{1-\alpha}\right\rfloor\\
& n= \left\lfloor\log\left( |\Delta|\right)^{\alpha}\right\rfloor,
\end{align*}
for some $\alpha\in\left[ \frac{1}{3},\frac{2}{3}\right[$. Then, $\OO_{\K_{n,K}}$ has discriminant satisfying: 
$$\log(\text{Disc}(\OO_{\K_{n,K}}))=\log( n^{n}K^{n-1}) = \log(|\Delta|) (1+o(1)).$$
If in addition we require that $n$ and $K$ be the largest prime numbers below their respective bounds such that: 
$$n^2 \nmid K^{n-1}-1,$$
then we meet the last restriction $\Z[\theta] = \OO_{\K_{n,K}}$.\\

We proceed by analogy with the approach of \cite{Enge} in the context of algebraic curves, where the authors examined curves of the form: 
$$\mathcal{C}:Y^n + X^d + f(X,Y),$$
such that any monomial $X^iY^j$ occuring in $f$ satisfies $ni+dj < nd$. The genus $g$ is assumed to tend to infinity and: 
\begin{align*}
n &\approx g^{\alpha} \\
d &\approx g^{1-\alpha}.
\end{align*}
The idea in \cite{Enge} is to look for functions $\phi(X,Y)\in\F_q[X,Y]$ satisfying: 
$$\deg_Y\phi\approx g^{\alpha-1/3}\ \ \text{and}\ \ \deg_X\phi\approx g^{2/3-\alpha},$$
with $\mathcal{N}(\phi)$ splitting into polynomials of degree bounded by $B = \log\left( L(1/3,\rho) \right)$ for some number $\rho$ determined in the complexity analysis. Each time such a decomposition occurs, the ideal $(\phi)$ is necessarily a product of primes belonging to the set $\mathcal{B}$ of the prime ideals of degree bounded by $B$: 
$$(\phi) = \prod_{\p_i\in \mathcal{B}} \p_i^{e_i}.$$
Such a decomposition of a principal ideal is called a \textit{relation}. In the following, we will also denote the vector $(e_i)$ itself a relation. Every time we find a relation, we add the row vector $(e_i)$ to a matrix $M\in\Z^{m\times N}$ called the \textit{relation matrix}, where $N := |\mathcal{B}|$, and $m\geq k$ is the number of relations collected. A linear algebra step is performed on this matrix. It consists in computing its Smith Normal Form, that is to say integers $d_1,\hdots,d_N$, with $d_N|d_{N-1}|\hdots |d_1$, such that there exist two unimodular matrices $U\in\Z^{m\times m}$ and $V\in\Z^{N\times N}$ satisfying:  

\[ M = U\left( 
   \begin{BMAT}(@)[1pt,1cm,1cm]{c}{c.c}
   \begin{BMAT}(e){cc}{c}
   \begin{BMAT}(e){ccc}{ccc}
d_1&       &      (0) \\
   & \ddots&       \\
 (0)  &       & d_k       
  \end{BMAT}  & (0)
  \end{BMAT} \\
  \begin{BMAT}[2pt,3cm,1cm]{c}{c} 
	(0)
  \end{BMAT}
  \end{BMAT}
   \right)V.  \]
The SNF of $M$ provides us with the group structure of the Jacobian of the curve $\mathcal{C}$. Indeed, if $\mathcal{L}_{\Z}$ is the lattice spanned by all the possible relations, and if $\mathcal{J}$ denotes the Jacobian of $\mathcal{C}$, then we have:
$$\mathcal{J} \simeq \Z^N / \mathcal{L}_{\Z}.$$
Providing $m$ is large enough to ensure that the rows of $M$ generate $\mathcal{L}_{\Z}$, we have:
$$\mathcal{J}\simeq \bigoplus_i \Z/d_i\Z.$$

In our context, we need the group structure of $\Cl(\OO_{\K})$, along with the regulator $R$, and a system of fundamental units of $\OO_{\K}$. The computation of the group structure of $\Cl(\OO_{\K})$ is done using methods similar to those used for the computation of the structure of $\mathcal{J}$. We look for relations of the form: 
$$(\phi) = \prod_i \p_i^{e_i},$$
where $\phi\in\K$, and where the $\p_i$ are prime ideals of norm bounded by $L(1/3,\rho)$. Every time we find such a relation, we add the row vector $(e_i)_{i\leq N}$ to the relation matrix denoted by $M_{\Z}\in\Z^{m\times N}$. To continue the analogy with \cite{Enge}, we require that $\phi$ be of the form:
$$\phi = A(\theta),$$ 
where $A\in\Z[X]$ of degree $k$. During the analysis, we will provide bounds on $k$ and on the coefficients of $A$, that delimit the search space. Providing the rows of $M_{\Z}$ generate the lattice $\mathcal{L}_{\Z}$ of all the possible row vectors $(e_i)_{i\leq N}\in\Z^N$ representing a relation, we have: 
$$\Cl(\OO_{\K})\simeq \Z^N/\mathcal{L}_{\Z} \simeq \bigoplus_{i\leq N} \Z/d_i\Z,$$
where the $d_i$ are the diagonal coefficients of the SNF of $M_{\Z}$. The main difference with the context of algebraic curves is the computation of $R$ and of a system of fundamental units. The group of units of $\OO_{\K}$ is of the form: 
$$U(\K)\simeq \mu(\K)\times\Z^r,$$
where $\mu(\K)$ is the multiplicative group of the roots of unity in $\OO_{\K}$. A system of fundamental units $(\gamma_i)$, $i\leq r$, is a set of elements of $\K$ satisfying: 
$$U(\K)\simeq \mu(\K)\times\left\langle \gamma_1\right\rangle \times\hdots\times\left\langle \gamma_r\right\rangle .$$
Once such a system is found, we use the logarithm map:
\[   \left.  \begin{array}{cccc}
        & \K & \longrightarrow & \R^{r+1}\\
      \text{Log}:  & \phi & \longmapsto & (\log|\phi|_1,\hdots,\log|\phi|_{r+1}), \end{array} \right.\] 
where the $|.|_j$ are the archimedian valuations on $\K$, to construct a matrix $A\in\R^{r\times(r+1)}$ whose rows are the vectors $\text{Log}(\phi_i)$, for $i\leq r$. The regulator is defined as the determinant of any $r\times r$ minor of $A$. To construct $A$ and a system of fundamental units, we augment the row vectors by columns containing the archimedian valuations, and add the row: 
$$(e_1,\hdots,e_k,\log|\phi|_1,\hdots,\log|\phi|_{r+1})\in\Z^N\times\R^{r+1}$$
to a relation matrix $M$ whenever a relation $(\phi) = \prod_i \p_i^{e_i}$ is found. A linear algebra step performed on $M$ provides us with the group structure, the regulator, and a system of fundamental units. It is described in detail in \textsection \ref{linear}.

\section{The relation matrix}

Let $\rho$ be a constant to be determined later, and $B$ a smoothness bound satisfying: 
$$B = \lceil L_{\Delta}(1/3,\rho)\rceil.$$ 
We define the factor base $\mathcal{B}$ as the set of all non inert prime ideals of norm bounded by $B$. This factor base has cardinality: 
$$N:=|\mathcal{B}| = L(1/3,\rho+o(1)).$$ 
In the following, we will need to test the smoothness of principal ideals of the form $(\phi)$, where $\phi = A(\theta)$ with $A\in\Z[X]$. We will use the well-known result that is proved in \cite{cohen}, Lemma 3.3.4: 
\begin{lemma}\label{lemma_norm}
 The norm of $\phi$ satisfies:
$$\mathcal{N}(\phi) = \text{Res}\ ( T(X) , A(X) ),$$
where $\text{Res}$ denotes the resultant.
\end{lemma}

Computing $\mathcal{N}(\phi)$ for $\phi\in\K$ allows us to decide whether $\phi$ is a product of prime ideals $\p\in\mathcal{B}$. Indeed, it suffices to check if $\mathcal{N}(\phi)\in\Z$ is $B$-smooth which can be done by trial division or the ECM method in polynomial time.
We assume that the coefficients $a_i$ of the polynomial $A$ have their logarithm bounded by an integer $a$, and that there exist two constants $\delta$ and $\nu$ to be determined later such that: 
\begin{align}
 a &\leq \left\lceil \delta \frac{\kappa\log|\Delta|/n}{(\log|\Delta|/\mathcal{M})^{1/3}} \right\rceil \label{bound_a}\\
 k &\leq \left\lceil \nu \frac{n}{(\log|\Delta|/\mathcal{M})^{1/3}} \right\rceil, \label{bound_k}
\end{align}
with $\mathcal{M}:=\log\log|\Delta|$. Using Lemma \ref{lemma_norm} and Hadamard's inequality, we deduce an upper bound on $\log\mathcal{N}(\phi)$:
\begin{align}
\log\mathcal{N}(\phi) &\leq na + dk + n\log k + k\log n \\
 & \leq \kappa\log\left( |\Delta|\right)^{2/3}\mathcal{M}^{1/3}(\delta + \nu + o(1) ). \label{bound_norm}
\end{align}
In the following, we will also need a bound on the real coefficients $\log|\phi|_i$  occuring in the relation matrix. By the following proposition, we derive a bound on the $\log|\theta|_i$ from the imposed bounds on the coefficients of $T$:

\begin{proposition}
Let $\sigma_i$ be the $n$ complex embeddings of $\K$ such that we have $T = \prod_i (X-\sigma_i(\theta))$, then the $\sigma_i(\theta)$ satisfy: 
$$\log(|\theta|_i) = \log(|\sigma_i(\theta)|) = O(\log\left( |\Delta|\right) ^{1-\alpha}).$$ 
\end{proposition}

\begin{proof}
Landau-Mignotte's theorem \cite{mignotte} states that if $D\mid T$ with $\deg D=m$, then the coefficients $d_j$ of $D$ satisfy:
$$|d_j|\leq 2^{m-1}(|T| + t_n),$$
where $|T|$ is the euclidian norm of the vector of the coefficients of $T$. Applying this to $D = X-\sigma_i(\theta)$ and $m=1$ allows us to obtain:
$$\log(|\theta|_i)\leq \log(|T|+t_n) \in O(\log\left( |\Delta|\right) ^{1-\alpha}).$$
\end{proof}

\begin{corollary}
With $\phi = A$, and $a$ and $k$ respectively bounded by (\ref{bound_a}) and (\ref{bound_k}), we have:
$$\log|\phi|_i \leq O(\log\left( |\Delta|\right) ^{2/3}\M^{1/3}).$$
\end{corollary}
To compute the probability for $\phi$ to be $\mathcal{B}$-smooth, we have to make the following assumption: 

\begin{heuristic}
We assume that $\mathcal{N}(\phi)$ behaves like a random number whose logarithm satisfies
$$\log(\mathcal{N}(\phi))\leq \iota:=\kappa\log\left( |\Delta|\right) ^{2/3}\mathcal{M}^{1/3}(\delta + \nu + o(1) ),$$
and whose distribution is given by the $\psi$ function of \cite{Canfield}.
\end{heuristic}
Consequently, computing the probability for a given $(\phi)$ to be $\mathcal{B}$-smooth boils down to computing the probability for a number whose logarithm is bounded by $\iota$ to be smooth with respect to prime numbers with logarithm bounded by 
$$\mu := \lceil \rho\log\left( |\Delta|\right) ^{1/3}\M^{2/3}\rceil.$$
Using \cite{Canfield}, and carrying out the same computation as in the proof of Theorem 1 of \cite{Enge}, one readily shows the following result on the probability of finding a relation:
\begin{proposition}\label{smoothness}
Let:
\begin{align*}
\iota&= \lfloor\log L(\zeta , c)\rfloor = \lfloor c\log\left( |\Delta|\right) ^{\zeta}\M^{1-\zeta}\rfloor \\
\mu&=  \lceil\log L(\beta,d)\rceil= \lceil d\log\left( |\Delta|\right) ^{\beta}\M^{1-\beta}\rceil,
\end{align*}
then we have: 
$$\frac{\psi(\iota,\mu)}{e^{\nu}}\geq L\left( \zeta-\beta,\frac{-c}{d}(\zeta-\beta)+o(1)\right), $$
where $\psi(\iota,\mu)$ denotes the cardinality of the set of integers $x$ such that $\log x\leq \iota$, and $x$ is smooth with respect to the set of prime numbers $p$ such that $\log p\leq \mu$. 
\end{proposition}

\section{The linear algebra phase}\label{linear}

In this section, we start with an overview of the linear algebra phase, then we address its complexity in \textsection \ref{LLL} and \textsection\ref{reg_computation}. We denote by $M$ the relation matrix whose rows lie in $\Z^N\times\R^{r+1}$, and by $M_{\Z}$ and $M_{\R}$ the matrices formed respectively by the first $N$ and the last $r+1$ columns of $M$. $M$ thus has the following shape: 
\[M=
\left( 
\begin{BMAT}[2pt,3cm,1cm]{c.c}{c} 
M_{\Z} & M_{\R}
\end{BMAT}
\right).  \]
To make sure we generate the full lattice of relations, we make the following assumption:
\begin{heuristic}\label{heuristic_dim}
We assume that there is a constant $K_1$ such that collecting $N+K_1r$ allows us to generate the full lattice of relations.
\end{heuristic}

In the following, we assume that Heuristic \ref{heuristic_dim} is satisfied. If this is not the case (which can be tested easily as we will see at the end of this section), we start all over again and construct another relation matrix. $M_{\R}$ contains rational approximations of the $\log |\phi_i|_j$ for $i\leq N+K_1r$ and $j\leq r+1$: the discussion of approximation issues when we add or multiply two real numbers is postponed to \textsection \ref{approximation}. As the rows of $M$ are assumed to generate the full lattice of the relations, the determinant of the lattice $\mathcal{L}_{\Z}$ spanned by the rows of $M_{\Z}$ gives us the class number $h(\OO_{\K})$, and its Smith Normal Form $\text{diag} (d_1,\hdots , d_N)$ gives us the decomposition
$$\Cl(\OO_{\K})\simeq \Z^N/\mathcal{L}_{\Z} \simeq\bigoplus_i \Z/d_i\Z.$$ 
On the other hand, we need to construct $r$ relations of the form 
$$(0,\hdots,0,\log|\gamma|_1,\hdots,\log|\gamma|_{r+1}),$$
along with the corresponding values of $\gamma$ (that are necessarily units), such that these relations generate the lattice $\mathcal{L}_{\R}$ of relations whose integer part contains only zero coefficients. 
To do this, we compute separately the Hermite Normal Form of $M_{\Z}$ and a basis $(\underline{u}_j)_{j\leq l}$ with $l\leq K_1r$ of the kernel of $M_{\Z}$. Then, we apply the $\underline{u}_j$ to $M_{\R}$, thus obtaining a matrix $A_{\R}\in\R^{l\times(r+1)}$ whose rows correspond to the archimedian valuations of units $(\beta_j)_{j\leq l}$. More details on this part of the algorithm are given in \textsection \ref{LLL}. To compute the regulator $R$, we need to find $r$ combinations of rows of $A_{\R}$, along with the corresponding units $(\gamma_i)_{i\leq r}$, that span the lattice of units $\mathcal{L}_{\R}$. This procedure is described in \textsection \ref{reg_computation}.  

At the end of the linear algebra phase, we have to check a posteriori that $N+K_1r$ relations were enough to generate $\mathcal{L}_{\Z}$ and $\mathcal{L}_{\R}$. The analytic class number formula provides a number $h^*$ computable in polynomial time satisfying:
$$h^* \leq h(\OO_{\K})R < 2h^*.$$
Before going into more details on the linear algebra phase, we recall the main steps of this process: 

%\begin{enumerate}
% \item Use LLL algorithm to compute $U$ and thus $A_{\Z}$ and $A_{\R}$,
% \item Compute the SNF of $A_{\Z}$ and deduce $h(\Delta)$ and the group structure of $\Cl(\OO_{\K})$,
% \item Find $r$ independent relations generating $\mathcal{L}_{\R}$ along with the corresponding units,
% \item Compute the determinant $R$ of $\mathcal{L}_{\R}$.
%\end{enumerate}
\begin{algorithm}[H]
\caption{Linear algebra phase}
\begin{algorithmic}[1]\label{linalg}
\REQUIRE $M$ 
\ENSURE $h(\OO_{\K})$, the structure of $\Cl(\OO_{\K})$, $R$, and a system of fundamental units
\STATE Compute the HNF of $M_{\Z}$.
\STATE Compute the SNF of $M_{\Z}$ and deduce $h(\OO_{\K})$ and the group structure of $\Cl(\OO_{\K})$.
\STATE Compute a basis $(\underline{u}_j)_{j\leq l}$ of $\ker M_{\Z}$ and deduce $A_{\R}$
\STATE Find $r$ independent relations generating $\mathcal{L}_{\R}$ along with the corresponding units.
\STATE Compute the determinant $R$ of $\mathcal{L}_{\R}$.
\STATE Compute $h^*$ and check if $h^* \leq h(\OO_{\K})R < 2h^*$. If not create another $M$ and go back to step one.
\end{algorithmic}
\end{algorithm}

\begin{notation}
In the following, $\textbf{r}^X_i$ denotes the row number $i$ of the matrix $X$.
\end{notation}

\subsection{Hermite and kernel basis computation}\label{LLL}

To obtain the matrix $A_{\R}$, we apply the kernel basis computation algorithm described in \cite{JacobsonHNF} to the rectangular matrix $M_{\Z}$. It provides $l\leq K_1r$ vectors $\underline{u}_j$ in $\Z^{N+K_1r}$ representing linear dependencies between the rows of $M_{\Z}$. Applying those linear combinations to the rows of $M$ yields $l$ relations with zero coefficients on the first $N$ coordinates. We denote by $\mathcal{L}_{\R}$ the lattice of the relations having only zeros on their first $N$ coordinates. As we assume Heuristic \ref{heuristic_dim}, these $l$ relations generate $\mathcal{L}_{\R}$. The last $r+1$ coordinates of each of the $l$ relations created this way are added as a row vector to the matrix $A_{\R}$. In addition, for every $\underline{u}_j$ of the form:
$$\underline{u}_j=(u_j^{(1)},\hdots,u_j^{(N+K_1r)}),$$
and for all $j\leq l$, the value $\beta_j=\prod_i \phi_i^{u_j^{(i)}}$ is the unit corresponding to the row 
$$\textbf{r}_j^{A_{\R}}=\sum_i u_j^{(i)}\textbf{r}_i^M.$$ 
As we will see in \textsection \ref{sub}, the coefficients $u_j^{(i)}$ are too large to allow us to compute directly $\prod_i \phi_i^{u_j^{(i)}}$ in subexponential time. We thus give the units $\beta_j$ in compact representation, that is to say by storing the $\underline{u}_j$.
It is proved in \cite{Arne} that the computation  of the $\underline{u}_j$ takes:
$$O(l^2N^3(\log N + \log |M_{\Z}|),$$
where $ |M_{\Z}| = \max_{i,j}\left\lbrace |M^{i,j}_{\Z}|\right\rbrace $. We need a bound on $\ |M_{\Z}|$ to express this complexity in terms of the size of the input: 

\begin{proposition}\label{bound_B}
$ |M_{\Z}|$ satisfies:
$$ |M_{\Z}| = O(\left( \log|\Delta|\right) ^{2/3}\left( \log\log|\Delta|\right) ^{1/3}).$$
\end{proposition}

\begin{proof}
We restricted ourselves to $\phi$ satisfying 
$$\log(\mathcal{N}(\phi))\leq \kappa\left( \log|\Delta|\right) ^{2/3}\left( \log\log|\Delta|\right) ^{1/3}(\delta+\nu + o(1)).$$
If $\mathcal{N}(\phi)=\prod_i\mathcal{N}(\p_i)^{e_i}$, then we clearly see that the vector $(e_i)$ having the largest coefficient under the previous constraint is the one where $e_1$ is maximal and all the others are set to zero, providing we set $\mathfrak{p}_1$ to the prime ideal of smallest norm. In that case, $e_1$ satisfies:
$$ e_1 = O(\left( \log|\Delta|\right) ^{2/3}\left( \log\log|\Delta|\right) ^{1/3}).$$
\end{proof}

\begin{corollary}
The complexity of the computation of the kernel basis of $M_{\Z}$ is bounded by:
$$O\left( L(1/3 , 3\rho + o(1)\right). $$
\end{corollary}

We use the HNF algorithm described in \cite{JacobsonHNF}. Its bit complexity is bounded by:
$$O\left( lN^3\left( \log N + \log|M_{\Z}|\right)^3 +N^5\left( \log N + \log|M_{\Z}|^2\right)   \right) .$$
This allows us to determine explicitly the expected time taken by the computation of the HNF and of the kernel basis of $M_{\Z}$ with respect to the size of the entries:

\begin{proposition}
The computation of the HNF and of the kernel basis of $M_{\Z}$ has bit complexity bounded by:
$$O\left( L(1/3 , 5\rho + o(1)\right). $$
\end{proposition}

In the following, we will need bounds on $|\underline{u}_j|$ and on $|A_{\R}|$. Direct application of the methods used in \cite{JacobsonHNF} leads to the following result:
\begin{lemma}\label{xbound}
$|\underline{u}_j|$ and $|A_{\R}|$ satisfy:
\begin{align*}
\log |\underline{u}_j| &= O\left( L(1/3 , \rho + o(1)\right) \\
\log  |A_{\R}| &= O\left( L(1/3 , \rho + o(1)\right).
\end{align*}
\end{lemma}

\subsection{The computation of $R$ and of the system of fundamental units}\label{reg_computation}

To compute the regulator and a system of fundamental units, we have to find a set of $r$ row vectors that span $\mathcal{L}_{\R}$. To do that, we take successive $r\times r$ determinants from submatrices extracted from $A_{\R}$, and we perform some elementary operations on the rows of $A_{\R}$. This procedure is described in Algorithm \ref{comp_Reg}, which was first introduced in \cite{cohen}, Algorithm 6.5.7. It makes use of the real GCD algorithm, which is also presented in \cite{cohen}, Algorithm 5.9.3. Given two multiples of the regulator $aR$ and $bR$, where $a$ and $b$ are integers, the real GCD algorithm outputs $dR$, where $d$ is the GCD of $a$ and $b$, under the assumption that $R > 0.2$. Algorithm \ref{comp_Reg} also calls the pre-computation step described in Algorithm \ref{indpt_rows}. This step, not presented in \cite{cohen}, is essential to ensure the validity of Algorithm \ref{comp_Reg}.

\begin{algorithm}[H]
\caption{Computation of the regulator and a system of fundamental units}
\begin{algorithmic}\label{comp_Reg}
\REQUIRE  $A_{\R}$ and the corresponding units $\beta_i$
\ENSURE  $R$ and a system of fundamental units
\STATE $R_1\leftarrow 0$
\STATE $i\leftarrow r-2$
\STATE Find $r$ linearly independent rows using Algorithm \ref{indpt_rows}
\WHILE {$i < l$}
\STATE Let $A$ be the matrix obtained by extracting any $r$ columns and rows $i-r +2$ to $i$ from $A_{\R}$.
\STATE $R_2\leftarrow \det A$
\STATE Using the real GCD algorithm, compute $u,v,R_3$ such that 
$$uR_1+vR_2=R_3$$
\STATE $R_1\leftarrow R_3$
\STATE $\gamma_i \leftarrow \beta_i^v\times \left. \beta_{i-r+1}\right. ^{(-1)^r u}$
\STATE $i\leftarrow i+1$ 
\ENDWHILE
\STATE $R\leftarrow R_1$
\end{algorithmic}
\end{algorithm}

\begin{algorithm}[H]
\caption{Search for $r$ independent rows}
\begin{algorithmic}\label{indpt_rows}
\REQUIRE  $A_{\R}$
\ENSURE A permutation of the rows of $A_{\R}$ such that the first $r$ are independent
\STATE $A_1 \leftarrow \textbf{r}^{A_{\R}}_1$
\STATE $i\leftarrow 1$
\FOR{$i=2$ to $r$}
\STATE $m\leftarrow i$
\STATE $\text{ret} \leftarrow 0$
\WHILE {$\text{ret} = 0$}
\STATE
\[A_{i}\leftarrow
\left( 
\begin{BMAT}[0.5pt,1cm,1.5cm]{c}{c.c} 
A_{i-1} \\
\textbf{r}^{A_{\R}}_m 
\end{BMAT}
\right).  \]
\IF{$\det (A_i^{t}A_i) = 0$}
\STATE $m\leftarrow m+1$
\ELSE
\STATE Swap $\textbf{r}^{A_{\R}}_i$ and $\textbf{r}^{A_{\R}}_m$ 
\STATE $\text{ret}\leftarrow 1$
\ENDIF
\ENDWHILE
\ENDFOR

\end{algorithmic}
\end{algorithm}

The main loop of Algorithm \ref{comp_Reg} ensures that the sub-lattice $\mathcal{L}_{\R}'$ of $\mathcal{L}_{\R}$ corresponding to the $\gamma_l$, for $i-(r-1)\leq l\leq i$, has determinant $R_3$.  Indeed, $\mathcal{L}_{\R}'$ is the sum of two sub-lattices of $\mathcal{L}_{\R}$ differing by a single element. The sign $(-1)^r$ is the signature of the permutation that is performed before this addition to make sure that $uR_1+vR_2=R_3$ holds by multilinearity of the determinant. The precomputation done with Algorithm \ref{indpt_rows} ensures that the first determinant computed is not null, which is essential for the completeness of Algorithm \ref{comp_Reg}. Whenever $\det (A_i^{t}A_i) \neq 0$, we have $i$ linearly independent rows.
%The complexity of this operation is bounded by
%$$\tilde{O}\left( r^3\log|A_{\R}|^3\right),$$
%where $\tilde{O}$ denotes the complexity in which the logarithmic factors are omitted. 
 
%In the worst case scenario we have to make $K_2r^2$ calls to this algorithm. Thus direct application of lemma \ref{xbound} proves the following result:
%\begin{proposition}
%The complexity of Algorithm \ref{indpt_rows} is bounded by:
%$$O\left( L(1/3 , 3\rho + o(1))\right) $$
%\end{proposition}
We postpone the computation of the complexity of Algorithms \ref{comp_Reg} and \ref{indpt_rows} to \textsection \ref{approximation}, where we calculate the precision we have to take for the rational approximations of the logarithms. In \textsection \ref{approximation}, we also ensure that this precision is accurate enough to enable us to decide whever $\det (A_i^{t}A_i) = 0$ or not. Algorithm \ref{real_GCD} describes the real GCD computation. Its presentation and correctness can be found in \cite{cohen}.

\begin{algorithm}[H]
\caption{Real GCD algorithm}
\begin{algorithmic}\label{real_GCD}
\REQUIRE $R_1 = aR$ and $R_2 = bR$ with $R > 0.1$, $R_1 > R_2$ and $a,b\in\Z$
\ENSURE $R_3 = dR$ and $u,v\in\Z$ such that $uR_1+vR_2 = R_3$
\STATE $u_0\leftarrow 1$, $v_0\leftarrow 0$ 
\STATE $u_1\leftarrow 0$, $v_1\leftarrow 1$
\WHILE {$R_2 > 0.1$}
\STATE $q\leftarrow \lfloor R_1/R_2\rfloor$, $r \leftarrow R_1 - R_2\lfloor R_1/R_2\rfloor$
\STATE $u_1\leftarrow u_0 - qu_1$
\STATE $v_1\leftarrow v_0 - qv_1$
\STATE $R_1 \leftarrow R_2$
\STATE $R_2 \leftarrow r$
\ENDWHILE
\STATE $R_3 \leftarrow R_1$
\STATE $u\leftarrow u_1$, $v\leftarrow v_1$
\end{algorithmic}
\end{algorithm}

\section{Approximation issues}\label{approximation}

The matrix $M_{\R}$ contains fixed point rational approximations $\hat{x_{ij}}$ of the logarithms of the units $x_{ij}:=\log|\phi_i|_j$. In this section, we discuss the precision of the computation of the regulator. In the following, we count the precision in bits. For example, we say that $\hat{x}$ is a rational approximation of $x\in\R$ with precision $q$ if $|\hat{x}-x|< 2^{-q}$. Let $q_0$ be the precision of the matrix $M_{\R}$. We have for $i\leq N+K_1r$ and $j\leq r+1$:
$$\hat{x_{ij}}=\sum_{k=-q_0}^{\lceil\log |x_{ij}| \rceil}2^ka_k^{ij},$$
where the $a_k^{ij}$ are the coefficients of the development of $x_{ij}$ as $\sum_{k=-\infty}^{\infty}2^k a_k^{ij}$. 
Before establishing the list of the steps where we might loose precision, we recall the following result that we will use to estimate the loss of precision whenever we add or multiply rational approximations:

\begin{lemma}\label{precision}
Let $\hat{x}$ and $\hat{y}$ be rational approximations of precision $q_1$ of respectively $x$ and $y$, and $u\in\Z$ such that $\lceil \log_2 u \rceil = q_2 < q_1$, then: 
\begin{itemize}
 \item $\hat{x}+\hat{y}$ is a rational approximation of $x+y$ of precision $q_1-1$.
 \item $u\hat{x}$ is a rational approximation of $ux$ of precision $q_1-q_2$.
 \item $\hat{x}\hat{y}$ is an approximation of $xy$ of precision $q_1 - \max\left\lbrace \log_2|x|,\log_2|y|\right\rbrace $.
\end{itemize}

\end{lemma}

$q_0$ is the precision taken for the approximation of the $\log|\phi_i|_j$. We set its value to:
$$q_0 := L(1/3 , 3\rho ).$$
The computation of the approximate value of each $\log|\phi_i|_j$ for $j\leq N+K_1r$ and $j\leq r+1$ takes $O(\text{M}(q_0)\log q_0)\in O\left( L(1/3 , 3\rho +o(1))\right) $ bit operations \cite{brent}. As we have to perform this computation
$$(r+1)(N+K_1r)\in O\left( L(1/3 , \rho + o(1)\right)$$
times, the time taken for the creation of $M_{\R}$ is bounded by $ O\left( L(1/3 , 3\rho +o(1))\right)$. Now, let us procede with the enumeration of the steps in the algorithm that deteriorate the precision. The first source of error is the computation of the coefficients of the matrix $A_{\R}$. Indeed, it contains rational approximations of 
$$\sum_{i=1}^{N+K_1r}u_j^{(i)}\log|\phi_i|_j,$$
for $j=1,\hdots,l$. The loss of precision is due to the multiplications by the $u_j^{(i)}$ and to the $N+K_1r$ additions. We deduce from Lemma \ref{xbound} the following proposition that gives us the loss of precision occuring in the computation of the coefficients of $A_{\R}$ with respect to the original precision taken during the construction of $M$:

\begin{proposition}
The computation of $\sum_i u_j^{(i)}\log|\phi_i|_j$ for $j=1,\hdots,r+1$, with precision $q'$, requires that the precision $q_0$ of the $\log|\phi_i|_j$ be: 
$$q' + N + K_1r + \max_{i,j}\left\lbrace \log_2|u_j^{(i)}|\right\rbrace .$$
Thus, the loss of precision during the computation of $A_{\R}$ is bounded by 
$$O(L(1/3,\rho+o(1))).$$
\end{proposition}

\begin{proof}
Multiplying $\log_2|\phi_i|_j$ by $u_i$ induces a loss of 
$$\log_2|u_i|\in O (L(1/3,\rho+o(1)))$$ 
bits of precision. Furthermore every addition induces the loss of one bit of precision. As we perform $N+K_1r = O(L(1/3,\rho+o(1)))$ of them, we thus lose another $N+K_1r$ bits of precision. Consequently, the total loss of precision is bounded from above by:
$$N + K_1r+ \max_{i,j}\left\lbrace \log_2|u_j^{(i)}|\right\rbrace \in O(L(1/3,\rho+o(1))).$$
\end{proof}

Once $A_{\R}$ is obtained, we need to compute successive $r\times r $ determinants extracted from this matrix. Every computation of such a determinant induces a loss of precision. The following proposition allows us to evaluate the loss of precision for one computation of an $r\times r$ determinant of a matrix $\hat{\Omega}$ extracted from $A_{\R}$. 

\begin{proposition}\label{loss_det}
The computation with precision $q'$ of the determinant of an $r\times r$ matrix $\hat{\Omega}$ extracted from $A_{\R}$, and which is a rational approximation of $\Omega\in\R^{r\times r}$, requires that 
$$ q = q' + (r/2+1)\log_2(r)\log_2\left( |\Omega|^{r-1}+1 \right),$$
where $q$ is the precision of the coefficients of $A_{\R}$, and $|\Omega|=\max_{i,j}|\Omega_{ij}|$. Thus, the loss of precision during the computation of the determinant of an $r\times r$ submatrix of $A_{\R}$ is bounded by $O\left( L(1/3,\rho+o(1))\right) $.
\end{proposition}

\begin{proof}
We know that $\Omega=(\omega_1,\hdots,\omega_{r})$ and $\hat{\Omega}=(\hat{\omega_1},\hdots,\hat{\omega_{r}})$ are $r\times r$ matrices with $r\leq n \in O( \log_2\left( |\Delta|\right) ^{\alpha})$ and $|\Omega-\hat{\Omega}|\leq 2^{-q}$, and furthermore, by lemma \ref{xbound}, $\Omega$ satisfies $\log_2|\Omega|\in O(L(1/3,\rho+o(1)))$. We have by multilinearity of the determinant and by Hadamard's inequality:
\begin{align*}
|\det\hat{\Omega} - \det{\Omega}| &= |\sum_{i=1}^{r} \det(\omega_1,\hdots,\omega_{i-1},\hat{\omega}_i-\omega_i,\hat{\omega}_{i+1},\hdots,\hat{\omega}_r )|\\
& \leq r^{r/2+1}(|\Omega|^{r-1}+1)2^{-q}.
\end{align*}
Thus, the loss of precision is of 
$$ (r/2+1)\log_2(r)\log_2\left( |\Omega|^{r-1}+1 \right)= O(L(1/3,\rho+o(1))).$$
\end{proof}

The last source of loss of precision is the series of multiplications and additions involved in the computation of the real GCD of two approximations of multiples of the regulator. The following proposition gives us this loss of precision during the successive real GCD computations in Algorithm \ref{comp_Reg}, knowing from Heuristic \ref{heuristic_dim} that the real GCD need not be called more than $K_1r$ times.

\begin{proposition}
If we have the determinants of the successive $r\times r$ matrices with precision $q$, then we can obtain the regulator with precision $q'$ providing 
$$q = q' + \frac{K_1r^3}{4}\log_2^2 (r) \log_2^2 |A_{\R}|.$$
Thus, the loss of precision  during the successive real GCD computations is bounded by $O\left( L(1/3,2\rho+o(1))\right)$.
\end{proposition}

\begin{proof}
Whenever we compute another determinant $R_2$ of an $r\times r$ matrix extracted from $A_{\R}$, we have to perform the step
$$R_1\leftarrow R_2-R_1\lfloor R_2/R_1\rfloor$$
at most $\log_2 R_2$ times to get the real GCD of $R_1$ and $R_2$, where $R_1$ is the previous approximation of the regulator $R$. We know that the coefficients of the submatrix whose determinant is $R_2$ have bit size bounded by $$\log_2|\Omega| \leq \log_2|A_{\R}| \in O(L(1/3,\rho+o(1))).$$
By Hadamard's inequality, we have:
$$\log_2 R_2\leq r/2\log_2 r \log_2 | A_{\R}| =  O ( L(1/3,\rho+o(1))),$$
which gives us an upper bound on the number of times we enter the main loop of the real GCD algorithm. Every multiplication $R_1\left\lfloor \frac{R_2}{R_1}\right\rfloor$ induces the loss of at most $\log_2 R_2$ bits of precision. Thus, the total loss of precision of one call to the real GCD algorithm is of: 
$$\frac{r^2}{4}\log_2^2 (r) \log_2^2 |A_{\R}|=O\left( L(1/3,2\rho+o(1))\right) .$$ 
As we know that Algorithm \ref{real_GCD} is called at most $K_1r$ times, the loss of precision after the $K_1r$ calls for the real GCD algorithm is still of $L(1/3,2\rho+o(1))$ bits. The last thing we have to do is to check the validity of the value $\left\lfloor \frac{R_2}{R_1}\right\rfloor$. Indeed if $R_2/R_1$ is close to an integer, then we risk to compute $\left\lfloor \frac{R_2}{R_1}\right\rfloor \pm 1$. Assume that $R_1 = k_1R$ and $R_2 = k_2R$, with $k_1 = k_1'd$, $k_2 = k_2'd$, and with $k_1'$ and $k_2'$ coprime. Theoretically, we have 
$$\lfloor R_2 / R_1 \rfloor = \lfloor k_2' / k_1'\rfloor ,$$
but we can obtain the wrong value if $k_2'\sim Kk_1'$ for some integer $K$, the worst case scenario being $k_2' = Kk_1' \pm 1$ (we cannot have $k_2' = Kk_1'$ since $k_1'$ and $k_2'$ are coprime). In that case, we have:
$$\left|\frac{k_2'}{k_1'} - K\right| \geq \frac{1}{k_1'}.$$
Thus, we need that the precision be at most of $2\log_2|k_1'|$. As the loss of precision encountered so far is in  $O\left( L(1/3,2\rho+o(1))\right)$, and as the original precision is in $O\left( L(1/3,3\rho+o(1))\right)$, the current precision of the value $\left\lfloor \frac{R_2}{R_1}\right\rfloor$ is still in $O\left( L(1/3,3\rho+o(1))\right)$. Furthermore, $\log_2 | k_1' | \leq  \log_2 R_1\leq L(1/3,\rho+o(1))$, so the condition is satisfied and the value of the quotient can be trusted.
\end{proof}

\begin{corollary}
The total loss of precision is of:
$$N + K_1r+ \max_{i,j}\left\lbrace \log_2|u_j^{(i)}|\right\rbrace + \frac{K_1r^3}{4}\log_2^2 (r) \log_2^2 | A_{\R}| \in O\left( L(1/3,2\rho+o(1))\right).$$
\end{corollary}

These considerations allow us to evaluate the complexity of Algorithm \ref{comp_Reg}. Indeed, it consists of at most $K_1r$ computations of the determinant of an $r\times r$ submatrix $\hat{\Omega}$ of $A_{\R}$. Let $I\subset [1,K_2r]$ and $J\subset [1,r+1]$ both be subsets of cardinality $r$ such that $\hat{\Omega} =: (\hat{y_{ij}})_{i\in I,j\in J}$. In addition, we define $A=(a_{ij})_{i\in I,j\in J}\in\Z^{r\times r}$ such that it satisfies:
$$\hat{y_{ij}}=\sum_{k=-q}^{\lceil\log_2 |y_{ij}| \rceil}2^ka_k^{ij}=:\frac{a_{ij}}{2^q}.$$
We thus have by multilinearity: 
$$\det \hat{\Omega} = \frac{\det A}{2^{rq}},$$ 
where $q\in O\left( L(1/3 , 3\rho + o(1)\right) $ is the precision of the coefficients of $A_{\R}$. Furthermore, the computation of $\det A$ takes $\tilde{O}(r^4\log_2|A|)$ bit operations (see \cite{Arne}), where $\tilde{O}$ denotes the complexity when we omit the logarithm factors. Therefore, the expected time for the computation of $\det A$ is in
$$\tilde{O}\left( r^4\left( q+\log_2|\hat{\Omega}|\right)\right),$$
since $\log_2|A|=\max_{ij}\left\lbrace \log_2 a_{ij}\right\rbrace \leq q+\log_2|\hat{\Omega}|$. As  $q$ is in $O\left( L(1/3,3\rho+o(1))\right)$, and as we know from Lemma \ref{xbound} that $\log_2|\hat{\Omega}|\in O\left( L(1/3,\rho+o(1))\right)$, we have the following result on the complexity of Algorithm \ref{comp_Reg}:
\begin{proposition}
The complexity of  Algorithm \ref{comp_Reg} lies in
$$O(L(1/3,3\rho+o(1))).$$
\end{proposition}

Now, let us check the validity and the complexity of Algorithm \ref{indpt_rows}. Given an $r\times i$ submatrix $A_i$ of $A_{\R}$, we want to determine whether its rows are approximations of independant rows. To do this, we compute $\det (A_i^tA_i)$ and decide whether this is the approximation of a zero determinant. We use Minkowski's bound, which states that
\begin{equation}\label{minkowski}
\sqrt{\det A_i^tA_i } \geq \left( \frac{\| b_1^{(i)}\|_2}{\sqrt r}\right)^{r},
\end{equation}
where $b_1^{(i)}$ is the non-zero vector of minimal length in the lattice spanned by the rows of $A_i$. For every $i$, $b_1^{(i)}$ is the logarithm vector of a unit. In \cite{pohst}, it is shown that for every unit $\epsilon$ that is not a root of unity, we have:
\begin{equation}\label{Pohst}
\left( \sum_i \log|\epsilon|_i^2\right)^{1/2} > \frac{21}{128}\frac{\log n}{n^2}.
\end{equation}
Therefore, we can prove the following proposition:
\begin{proposition}
The precision $q_0=L(1/3,3\rho)$ is accurate enough to ensure the validity of Algorithm \ref{indpt_rows}, whose complexity is in
$$O(L(1/3,3\rho+o(1))).$$
\end{proposition}
\begin{proof}
First, we calculate the precision of the value $\det (A_i^tA_i)$. The coefficients $c^{(i)}_{kl}$  ($k,l\leq i$) of $A_i^tA_i$ are given by:
$$c^{(i)}_{kl}=\sum_{h\leq r} a^{(i)}_{kh} a^{(i)}_{lh},$$
where the $a^{(i)}_{kl}$  ($k\leq i$, $l\leq r$) are the coefficients of $A_i$. We know from Lemma \ref{xbound} that the coefficients of $A_i$ have bit size in $O\left( L(1/3,\rho + o(1) ) \right)$, thus, using Lemma~\ref{precision}, we prove that the precision of $c^{(i)}_{kl}$  ($k,l\leq i$) is still in $O(L(1/3,3\rho+o(1)))$. Using the same techniques as in Proposition \ref{loss_det}, we prove that the loss of precision we encounter during the computation of $\det (A_i^tA_i)$ is of
$$(i/2+1)\log_2(i)\log_2\left( |A_i^tA_i|^{i-1}+1 \right).$$
As $\log_2 |A_i^tA_i|\in O\left( L(1/3,2\rho + o(1) )\right)$, this loss of precision is in  $ O\left( L(1/3,2\rho + o(1) )\right)$ as well. We thus have the value of $\det (A_i^tA_i)$ with a precision $q$ satisfying: 
$$q\in O\left( L(1/3,3\rho + o(1) )\right).$$
On the other hand, we have a lower bound on the value of $\det (A_i^tA_i)$ from the combination of \eqref{minkowski} and \eqref{Pohst} in the case where $A_i$ contains approximations of independent rows:
$$\det (A_i^tA_i) \geq \left( \frac{21}{128}\right)^{2r}\frac{1}{r^r}\left( \frac{\log n}{n^2}\right)^{2r}.$$
If $\det (A_i^tA_i) \leq 1/2^q$, then it might equal zero, otherwise it is necessarily the approximation of a strictly non-zero determinant. Furthermore, the bound on $\det (A_i^tA_i)$ satisfies:
$$\left|\log\left[ \left( \frac{21}{128}\right)^{2r}\frac{1}{r^r}\left( \frac{\log n}{n^2}\right)^{2r}\right]\right|\leq n\log (n)(1+o(1))\ll q  .$$
We can thus conclude that if $\det (A_i^tA_i) \leq 1/2^q$, then the rows of $A_i$ are necessarily dependent.
\end{proof}
This allows us to state the following proposition:
\begin{proposition}
The complexity of the computation of $R$ and of the system of fundamental units lies in
$$O(L(1/3,3\rho+o(1))).$$
In addition, we know the value of $R$ with a precision:
$$q_R\in O\left( L(1/3 , 3\rho + o(1)\right).$$
\end{proposition}

\section{Subexponentiality}\label{sub}

In this section, we show that we achieve a subexponential complexity for the overall running time of the algorithm. Direct application of Proposition \ref{smoothness} with the parameters
\begin{align*}
&\beta = \frac{1}{3},\ d = \rho \\
&\zeta = \frac{2}{3},\ c = \kappa(\delta+\nu+o(1)),
\end{align*}
shows that the expected number of trials to obtain a relation is at most 
%$$L\left( \frac{1}{3},\frac{\tau}{\rho}+o(1)\right).$$
$$L\left( 1/3,\frac{\kappa(\nu+\delta)}{3\rho}+o(1)\right).$$
We know that the factor base has size $N\in O\left( L(1/3,\rho)\right) $, thus the complexity of the search for $N+K_1r$ relations is bounded by: 
$$L\left( 1/3,\frac{\kappa(\nu+\delta)}{3\rho}+\rho+o(1)\right).$$
The number of $\phi$ in the search space is in $O\left( L(1/3) , \nu\delta\kappa\right)$. We thus have the following constraint on the parameters:
\begin{equation}\label{constr1}
\nu\delta\kappa = \frac{\kappa(\nu+\delta)}{3\rho}+\rho.
\end{equation}
We can prove that the strategy minimizing the overall time is the one where the relation collection and the linear algebra take the same time. As the complexity of the linear algebra is dominated by the HNF computation which lies is $O(L(1/3,5\rho+o(1)))$, we thus have the additional constraint:
\begin{equation}\label{constr2}
\kappa\nu\delta = 5\rho. 
\end{equation}
From \eqref{constr1} and \eqref{constr2}, we obtain: 
\begin{align*}
\nu\delta &= \frac{5\rho}{\kappa} \\
\nu + \delta &= \frac{12b^2}{\kappa}.
\end{align*}
Thus, $\delta$ and $\nu$ are roots of the polynomial:
$$X^2 - \frac{24\rho^2}{\kappa}X + \frac{5\rho}{\kappa}.$$
These roots exist providing we have:
$$\rho\geq \sqrt[3]{\frac{5\kappa}{144}}.$$
The optimal choice is to minimize $\rho$, thus fixing the parameters $\delta$ and $\nu$:
$$\delta = \nu = \sqrt{\frac{5\rho}{\kappa}} = \sqrt[6]{\frac{625}{144\kappa^2}}.$$
The total running time becomes $L(1/3 , 5\rho + o(1))$, with:
$$\rho= \sqrt[3]{\frac{5\kappa}{144}}.$$

\section*{Acknowledgments}

The author thanks Andreas Enge for his support, the fruitful discussions we had, and his careful reading of this article. He also thanks Steven Galbraith for the original suggestion of adapting the $L(1/3)$ algorithm of \cite{Enge} to the context of number fields, and Michael Pohst for pointing out \cite{pohst}.

\end{document}